\documentclass[letterpaper,10pt,conference]{ieeeconf}  

\IEEEoverridecommandlockouts                              
\usepackage{amsmath}
\usepackage{amssymb}
\usepackage{mathptmx}
\usepackage{mathtools}
\usepackage[utf8]{inputenc}
\usepackage{color}
\usepackage{esint}

\usepackage[sort,noadjust]{cite}

\mathtoolsset{showonlyrefs}

\newtheorem{thm}{\protect\theoremname}[section]
\newtheorem{defn}[thm]{\protect\definitionname}
\newtheorem{rem}[thm]{\protect\remarkname}
\newtheorem{lem}[thm]{\protect\lemmaname}
\newtheorem{prop}[thm]{\protect\propositionname}
\newtheorem{cor}[thm]{Corollary}
\newtheorem{asm}[thm]{Assumption}

\newcommand{\norm}[1]{\lVert #1 \rVert}

\newcommand{\abs}[1]{| #1 |}

\DeclareMathOperator{\diag}{diag}

\DeclareSymbolFont{bbold}{U}{bbold}{m}{n}
\DeclareSymbolFontAlphabet{\mathbbold}{bbold}
\newcommand{\onev}{\mathbbold{1}}

\usepackage{calrsfs}
\DeclareMathAlphabet{\pazocal}{OMS}{zplm}{m}{n}
\renewcommand{\mathcal}[1]{\pazocal{#1}}

\title{\LARGE \bf Dynamic Analysis of Bet-Hedging Strategies as a Protection Mechanism against Environmental Fluctuations}

\author{Masaki Ogura$^{1}$, Masashi Wakaiki$^{2}$, and Victor M.~Preciado$^{1}$
\thanks{$^{1}$M.~Ogura and V.~M.~Preciado are with the Department of Electrical and Systems Engineering, University of Pennsylvania, PA 19104, USA. 
{\tt\small \{ogura,preciado\}@seas.upenn.edu}}%
\thanks{$^{2}$M.~Wakaiki is with the Department of Electrical and Electronic Engineering, Chiba University, 
1-33 Yayoi-cho, Inage-ku, Chiba 263-8522, Japan. {\tt \small  wakaiki@chiba-u.jp}}%
\thanks{This work was supported in part by the NSF under grants CNS-1302222 and IIS-1447470.}%
}

\providecommand{\definitionname}{Definition}
\providecommand{\lemmaname}{Lemma}
\providecommand{\propositionname}{Proposition}
\providecommand{\remarkname}{Remark}
\providecommand{\theoremname}{Theorem}

\newenvironment{proofof}[1]{
\begin{proof}}{\end{proof}
}

\usepackage{mathtools}
\mathtoolsset{centercolon}

\newcommand{\afterequation}{\vskip 3pt}

\allowdisplaybreaks[4]

\usepackage[all]{xy}

\begin{document}

\maketitle
\thispagestyle{empty}
\pagestyle{empty}

\begin{abstract}
In order to increase their robustness against environmental fluctuations, many
biological populations have developed \emph{bet-hedging} mechanisms in which the
population `bets' against the presence of prolonged favorable environmental
conditions by having a few individual behaving as if they sensed a threatening
or stressful environment. As a result, the population (as a whole) increases its
chances of surviving environmental fluctuations in the long term, while
sacrificing short-term performance. In this paper, we propose a theoretical
framework, based on Markov jump linear systems, to model and evaluate the
performance of bet-hedging strategies in the presence of stochastic
fluctuations. We illustrate our results using numerical simulations.
\end{abstract}

\section{Introduction}
Biological populations, such as bacterial colonies, are subject to multiple sources of environmental fluctuations, from regular cycles of daily light and temperature to irregular
fluctuations of nutrients and pH levels~\cite{Kussell2005,Acar2008,Saether2015}. In order to increase their robustness against environmental fluctuations, many biological systems have developed \emph{bet}\nobreakdash-\emph{hedging} mechanisms~\cite{Saether2015,Seger1987} in which the population `bets' against the presence of prolonged favorable environmental conditions by having a few individual behaving as if they sensed a threatening or stressful environment.  For example, in bacterial colonies, some bacteria may stochastically switch into a state of slow metabolic state, in which they are more robust against pH fluctuations. As a result, the population (as a whole) increases its chances of surviving pH fluctuations in the long term, while sacrificing short-term performance. Similar bet-hedging strategies can be found in many other biological systems, such as the lysis-lysogeny switch of bacteriophage $\lambda$ \cite{Oppenheim2005}, delayed germination in plants~\cite{Gremer2014}, and phenotypic variations in bacteria~\cite{Woude2004}.

In this paper, we pay special attention to a particular type of bet-hedging
mechanism based on introducing delays in the function of a few individuals in
the population. For example, in the case of cell populations, the presence of
time-delays in some basic patterns of cell proliferation can significantly
improve the overall population fitness~\cite{Baker1998}. Similarly, delayed
germination in plant populations~\cite{Gremer2014} and delayed disease
activation of viruses~\cite{Stumpf2002} have also been reported as bet-hedging
strategies in biological systems. In the current literature, the performance of
bet-hedging strategies is evaluated using either extensive numerical
simulations, or overly simplistic assumptions. Based on numerical simulations,
the authors in \cite{Belete2015} found the optimal rates of adaptation (e.g.,
the rate at which bacteria switch into a slow metabolic state) to maximize the
growth rate of cell populations. Based on overly simplistic assumptions,
analytical calculations of growth rates of phenotypically heterogeneous
populations are performed by assuming that environmental fluctuations are either
slow enough~\cite{Kussell2005}, fast enough~\cite{Muller2013a}, or
periodic~\cite{Gaal2010}. Although the works mentioned above provide intuitive
explanations about the effects of bet-hedging strategies, there is still a lack
of a solid mathematical framework for the evaluation of bet-hedging strategies
under complex environmental fluctuations.

The aim of this paper is to present a rigorous and tractable framework to
quantify the growth rates of cell populations using bet-hedging strategies involving
time-delays. Building on the models in the literature
\cite{Thattai2004,Kussell2005,Belete2015}, we introduce a population model in
terms of positive Markov jump linear systems~\cite{Ogura2013f} with delays.
Among various types of delays, we specifically focus on those in proliferation
(i.e., in the state variables) and in adaptation to environmental changes (i.e.,
in the switching signals). In the former case, we show that the growth rate of
a population exhibiting both point and distributed delays is upper-bounded
by the maximum real eigenvalue of a particular Metzler matrix. In the
latter case, we consider stochastic delays in adaptation to environmental
fluctuation and show that the growth rate coincides with the maximum real
eigenvalue of a Metzler matrix. The proposed framework can also be used to study both point and distributed delays in the state variables in a unified manner, whereas these delays have been studied separately in the literature~\cite{Qi2015a,Jiao2015}.

This paper is organized as follows. After presenting the notation in Section~\ref{sec:Bet-Hedging-Populations-under}, we introduce several linear
growth models of bet-hedging populations involving time delays. Then, in
Section~\ref{sec:delayedProlife}, we derive an upper bound on the growth rates for
the case of delayed proliferation. Section~\ref{sec:delayedResponse} shifts our
focus to delayed adaptation and shows that the growth rate equals the maximum
eigenvalue of a Metzler matrix. Numerical simulations are presented in
Section~\ref{sec:Numerical-Simulations}.

\subsection{Mathematical Preliminaries}

We denote by $\mathbb{R}_+$ the set of nonnegative real numbers. The $1$-norm of
$x\in\mathbb{R}^{n}$ is defined by $\norm{x}=\sum_{i=1}^{n}\abs{x_i}$. The
symbol~$\onev_{n}$ denotes the column vector of length $n$ whose entries are all
one. By $u_i$, we denote the $i$-th canonical basis vector in~$\mathbb{R}^n$.
Let $U_{ij} = u_iu_j^\top$. We say that a matrix is nonnegative if its entries
are nonnegative. A square matrix is said to be \emph{Metzler} if its
off-diagonal entries are nonnegative. We say that a matrix is Hurwitz stable if
all of its eigenvalues have negative real parts. The Kronecker
product~\cite{Brewer1978} of matrices is denoted by $\otimes$. It is
known~\cite{Brewer1978} that, if the standard product of matrices $AB$ and $CD$
are well-defined, then
\begin{equation}\label{eq:ABCD}
(AB)\otimes (CD) = (A\otimes B)(C\otimes D).
\end{equation}
For a closed interval~$[a,b]$, the space~$C([a, b],\mathbb{R}^{n}_+)$ is
defined as the set of $\mathbb{R}^{n}_+$-valued continuous functions on
$[a,b]$ equipped with the norm $\norm{x} = \int_{a}^{b}\norm{x(t)}\,dt$.

Let $A_0$ be an $n\times n$ Metzler matrix, $A_1$, $\dotsc$, $A_m$ be $n\times
n$ nonnegative matrices, $h_1$, $\dotsc$, $h_m$ be nonnegative constants, and
$B\colon [0, \infty) \to \mathbb{R}_+^{n\times n}$ be a continuous function
having a compact support. Consider the positive linear system with delay
\cite{Ngoc2013}:
\begin{equation} \label{eq:LTIDelaysys}
\dfrac{dx}{dt}=A_0x(t)+ \sum_{i=1}^m A_ix(t-h_i)+(B*x)(t),
\end{equation}
where $*$ denotes the convolution. Let $T$ be the minimum number such that
$T\geq h_i$ for every $i\in \{1, \dotsc, m\}$ and the interval~$[0, T]$ contains
the support of $B$. We set the initial condition of
\eqref{eq:LTIDelaysys} as $x|_{[-T, 0]} = \phi$ for a function~$\phi \in
C([-T,0], \mathbb{R}^{n}_+)$. We say that the system~\eqref{eq:LTIDelaysys} is
\emph{exponentially stable} if there exist $C> 0$ and $\lambda < 0$ such that
$\norm{x(t)}\leq C e^{\lambda t}\norm{\phi}$ for all $\phi$ and $t\geq 0$. The
following stability characterization is given in \cite{Ngoc2013}:

\begin{prop}[{\cite[Theorem~III.1]{Ngoc2013}}] \label{prop:Ngoc}
The system~\eqref{eq:LTIDelaysys} is exponentially stable if and only if the
matrix $A_0+\sum_{i=1}^m A_i+\int_{0}^{\infty}B(\tau)\,d\tau$ is Hurwitz stable.
\end{prop}

We also give a review on random variables and stochastic processes. Let
$(\Omega,\mathcal{M},P)$ be a probability space. For an integrable random
variable $X$ on $\Omega$, its expected value is denoted by $E[X]$. If
$\mathcal{M}_{1}\subset\mathcal{M}$ is a $\sigma$-algebra, then
$E[X\mid\mathcal{M}_{1}]$ denotes the conditional expectation of $X$ given
$\mathcal{M}_{1}$. It is well known (see, e.g., \cite{Borkar1995}) that
\begin{equation} \label{eq:filtering}
E[X] = E[ E[X\mid\mathcal M_1] ]. 
\end{equation}
Let $f, g\colon \mathbb{R}\to\mathbb{R}^n$ be continuous
functions having their compact supports in $[0, \infty)$. Let $N_1$, $\dotsc$,
$N_m$ denote Poisson counters. We say that a left-continuous function $x$ taking
values in $\mathbb{R}^n$ is a solution of the stochastic differential equation
\begin{equation}
dx=(f*x)\,dt+ \sum_{i=1}^m (g_i*x)\,dN_i \label{eq:SIE}
\end{equation}
if we have $dx/dt=(f*x)(t)$ when no jump occurs at time~$t$, and
$x(t^{+})=x(t^{-})+(g_i*x)(t^{-})$ when $N_i$ jumps at time~$t$. For the sake of
completeness, we state the It\^o rule for stochastic differential equations with
Poisson jumps (see, e.g., \cite{Brockett2009}):

\begin{lem}\label{lem:Ito}
Assume that $x$ follows the stochastic differential equation~\eqref{eq:SIE}. Let
$\psi\colon\mathbb{R^{\mathit{n}}\to\mathbb{R}^{\mathit{m}}}$ be a
differentiable function and define $y(t)=\psi(x(t))$ for every $t \geq 0$. Then,
$y$ follows the stochastic differential equation $dy = ({d\psi}/{dx})(f*x)\,dt +
\sum_{i=1}^{m}\left(\psi(x+(g_i*x))-\psi(x)\right)\,dN_i$.
\end{lem}

\section{Delayed Models for Bet-Hedging Populations}
\label{sec:Bet-Hedging-Populations-under}

The aim of this section is to introduce models of bet-hedging populations
involving delays. We first review the delay-free population model given in
\cite{Kussell2005,Belete2015}. Building on this model, we then introduce two
models of bet-hedging populations involving delays in their state variable and
mode signals, respectively.

Let us consider a biological population growing in an environment fluctuating
among $n$ different environment types. The fluctuation is
modeled~\cite{Kussell2005} by a time-homogeneous Markov process
$\epsilon=\{\epsilon(t)\}_{t\geq0}$ taking values in $\{1,\dotsc,n\}$ and having
the infinitesimal generator $\Pi=[\pi_{ij}]_{i,j}\in\mathbb{R}{}^{n\times n}$.
Therefore, the transition probability of the environment is given by
\begin{equation*}
P(\epsilon(t+h) = j \mid \epsilon(t) = i) = 
\begin{cases}
\pi_{ij} h + o(h),\ \ i\neq j,\\
1+\pi_{ii} h + o(h),\ \ i= j, 
\end{cases}
\end{equation*}
where $o(h)/h\to 0$ as $h\to 0$. Each individual in the population can exhibit
one of $n$ different phenotypes $1$, $\dotsc$, $n$. We assume that the
population having phenotype $k$ grows with the instantaneous rate $g^k_i\geq 0$
under environment $i$. The phenotypes of individuals are assumed to dynamically
change, and the rate of switch from phenotype $k$ to phenotype $\ell$ under
environment $i$ is denoted by $\omega^{k\ell}_i\geq 0$. We define $\omega_i^{kk}
= -\sum_{\ell\neq k}\omega^{k\ell}_i$. Let $x_k(t)$ denote the number of
individuals having phenotype $k$ at time~$t$. Then, the growth of the population
can be modeled~\cite{Kussell2005,Belete2015} by the differential equations
\begin{equation}\label{eq:dx_k/dt:withoutdelay}
\Sigma_0 : \frac{dx_k}{dt}=g^k_{\epsilon(t)}x_k(t)
+\sum_{\ell= 1}^n \omega^{\ell k}_{\epsilon(t)}x_{\ell}(t),\ k=1, \dotsc, n.
\end{equation}
Building on this model, we below introduce two growth models of population
involving delays.

\subsection{Delayed Proliferation}

We first consider a dynamic model of bet-hedging consisting in introducing delays in state variables of the population. Such delays, which can arise from delayed
bet-hedging~\cite{Stumpf2002} or delayed proliferation~\cite{Baker1998}, make
the derivative $dx_k/dt$ depend not only on the size of the current populations, but
also on their past values. To deal with this case, we extend the basic
model~$\Sigma_0$ as follows
\begin{equation*}
\begin{multlined}[.95\linewidth]
\Sigma_1:
\frac{dx_k}{dt}=g^k_{\epsilon(t)} x_k(t)
+
\sum_{\ell=1}^n \omega^{\ell k}_{\epsilon(t)}x_{\ell}(t) 
+
p^k_{\epsilon(t)}x_k(t-d_{\epsilon(t)}^k)
+
\\
\int_{0}^{\infty}q^k_{\epsilon(t)}(\tau)x_k(t-\tau)\,d\tau, 
\  k=1,\dotsc,n.
\end{multlined}
\end{equation*}
It is naturally assumed that $p_k^i$ is a nonnegative number and $q_k^i$ is a
nonnegative function having a finite support in $[0, \infty)$ for all $k, i\in
\{1, \dotsc, n\}$. We specify the initial condition of the system $\Sigma_1$ by
\begin{equation*}
x_k|_{[-T,0]} = \phi_k \in C([-T,0], \mathbb{R}_+),\ k=1,\dotsc, n,
\end{equation*}
where $T$ is the minimum number such that $d_i^k\leq T$ and $[0,T]$ contains the
support of function $q_i^k$ for all $k, i\in \{1, \dotsc, n\}$.

We define the growth rate of the  model $\Sigma_1$ as follows:

\begin{defn}
For $\lambda \in \mathbb{R}$, we say that $\Sigma_1$ is
\emph{$\lambda$-exponentially stable} if there exists $C>0$ such that
$E[\sum_{k=1}^n x_k(t)] \leq C e^{\lambda t} \sum_{k=1}^n \norm{\phi_k}$ for all
$\phi_1$, $\dotsc$, $\phi_n$, and $\epsilon(0)\in \{1, \dotsc, n\}$. We define
the \emph{growth rate} of $\Sigma_1$ as the infimum of $\lambda$ such that
$\Sigma_1$ is $\lambda$-exponentially stable. If the growth rate of $\Sigma_1$
is negative, then we say that $\Sigma_1$ is \emph{exponentially stable}.
\end{defn}

\subsection{Delayed Adaptation}

Another type of delay in bet-hedging populations can be present in the
adaptation of the population to environmental fluctuations. In this case,
each individuals' information~$\sigma$ about the environment (on which their
adaptation is based on) does not necessarily coincide with the
environment~$\epsilon$ due to delays. This implies that, mathematically speaking, there
exists a nonnegative stochastic process~$h = \{h(t)\}_{t\geq 0}$ such that
$\sigma(t) = \epsilon(t - h(t))$. We assume that the growth rates depends on the
environmental variable~$\epsilon$, while the adaptation rates between phenotypes
depend on the delayed information~$\sigma$. In this situation, the basic
population model~$\Sigma_0$ has to be rewritten as
\begin{equation}\label{eq:dx_k/dt:resopsnedelay}
\Sigma_2 : 
\frac{dx_k}{dt}
=
g^k_{\epsilon(t)} x_k(t)
+
\sum_{\ell=1}^n\omega^{\ell k}_{\sigma(t)}x_{\ell}(t),\ 
k=1, \dotsc, n.
\end{equation}
We postpone the detailed description of the delay process~$h$ as well as the
definition of the growth rates to Section~\ref{sec:delayedResponse}.

\section{Growth Rate with Delayed Proliferation}\label{sec:delayedProlife}

The aim of this section is to prove the following theorem, which enables us to
find an upper bound of the growth rate of $\Sigma_1$:

\begin{thm}\label{thm:rate}
Let $\lambda \in \mathbb{R}$ be arbitrary. For all $i, j\in \{1, \dotsc, n \}$
and $t\geq 0$, define $f^{(\lambda)}_{ij} = p_i^je^{-\lambda d_i^j}  ((U_{ji}
e^{\Pi^\top d_i^j})\otimes u_j^\top)$ and $g^{(\lambda)}_{ij}(t) =
q_i^j(t)e^{-\lambda t}  ((U_{ji} e^{\Pi^\top t})\otimes u_j^\top)$. Let
\begin{equation}
\begin{aligned}
\bar A_{ij}^{(\lambda)} &= u_i\otimes f_{ij}^{(\lambda)}, \\
\bar B_{ij}^{(\lambda)}(t) &= u_i\otimes g_{ij}^{(\lambda)}(t). 
\end{aligned}
\end{equation}
Then, the growth rate of $\Sigma_1$ is less than $\lambda$ if the matrix
\begin{equation}
\bar T^{(\lambda)}
=
\bar A_0-\lambda I 
+
\sum_{i,j=1}^n \bar A_{ij}^{(\lambda)}
+
\sum_{i,j=1}^n\int_{0}^{\infty}\bar B_{ij}^{(\lambda)}(t)\,dt
\end{equation}
is Hurwitz stable. 
\end{thm}

\begin{rem}
We can use Theorem~\ref{thm:rate} and a bisection search to find the suboptimal
upper bound on the growth rates. We also remark that, when $\Sigma_1$ has no
delay, the sufficient condition in Theorem~\ref{thm:rate} is also necessary by
\cite[Theorem~5.1]{Ogura2013f}.
\end{rem}

The rest of this section is devoted to the proof of Theorem~\ref{thm:rate}. We
first introduce a vectorial representation of $\Sigma_1$. For each $i \in \{1,
\dotsc, n\}$, define $G_i = \diag(g^1_i, \dotsc, g^n_i)$, $\Omega_i =
[\omega^{k\ell}_i]_{k,\ell}$, $A_i = G_i + \Omega_i^\top$, and $Q_i =
\diag(q^1_i, \dotsc, q^n_i)$. Let
\begin{equation}
x = \begin{bmatrix}
x_1\\\vdots\\x_n
\end{bmatrix}, 
\ 
(\mathcal P_ix)(t) = \begin{bmatrix}
p_i^1 x_1(t-d_i^1)\\\vdots\\ p_i^n x_n(t-d_i^n)
\end{bmatrix},
\ 
\phi = \begin{bmatrix}
\phi_1\\\vdots\\\phi_n
\end{bmatrix}.  
\end{equation}
We can then write $\Sigma_1$ in a vector form
as
\begin{equation}\label{eq:dx/dt}
\Sigma_1 : \frac{dx}{dt} = A_{\epsilon(t)} x(t)
+(\mathcal P_{\epsilon(t)}x)(t) + (Q_{\epsilon(t)}*x)(t)
\end{equation}
with initial condition $x|_{[-T,0]} = \phi$. Let us also introduce the vectorial
representation $\eta =\{\eta(t)\}_{t \geq 0}$ for the environmental variable
$\epsilon$ by $\eta_i(t) = 1$ if $\epsilon(t) = i$ and $\eta_i(t) = 0$
otherwise. Notice that $\eta(t)=u_{\epsilon(t)}$.

In what follows, instead of directly dealing with the process~$x(t)$, we shall
study the auxiliary processes given by \cite{Ogura2013f}
\[
z(t)=\eta(t)\otimes x(t),\ \zeta(t)=E[z(t)],\ t\geq 0. 
\]
Notice that neither $z(t)$ nor $\zeta(t)$ is defined when $t<0$ because
$\eta(t)$ is defined only for $t\geq 0$. The next lemma shows that these
auxiliary processes preserve the norm of $x(t)$: 

\begin{lem}\label{lem:presevrve}
$\norm{\zeta(t)} = E[\norm{x(t)}]$ for every $t\geq 0$.
\end{lem}

\begin{proof}
Notice that, if $x\in \mathbb{R}^n$ is nonnegative, then $\norm x = \onev_n^\top
x$. Therefore, since $\zeta(t)\geq 0$ and $x(t)\geq 0$, we have $\norm{\zeta(t)}
= \onev_{n^2}^\top E[\zeta(t)] = E[(\onev_n\eta(t)) (\onev_n x(t))] = E[\onev_n
x(t)] = E[\norm{x(t)}]$.
\end{proof}

The following lemma plays an important role in the proof of the main result:

\begin{lem}\label{lem:key}
For all $i\in\{1, \dotsc, n\}$, $h \in [0, T]$, and $t\geq T$, we have
$E[\eta_i(t)x(t-h)]=((u_i^{\top}e^{\Pi^{\top}h})\otimes I_{n}) \zeta(t-h)$.
\end{lem}

\begin{proof}
Equation \eqref{eq:filtering} shows that
\begin{equation}
\begin{aligned} 
E[\eta_i(t)x(t-h)]
&=  E\bigl[E[\eta_i(t)x(t-h)\mid\eta(t-h)]\bigr]\\
&=  E\bigl[E[\eta_i(t)\mid\eta(t-h)]\,x(t-h)\bigr].
\end{aligned} \label{eq:Iusedfilterign}
\end{equation}
Since $\eta_i=u_i^{\top}\eta$, we can show $E[\eta_i(t)\mid\eta(t-h)]
=E[u_i^{\top}\eta(t)\mid\eta(t-h)] =u_i^{\top}e^{\Pi^{\top}h}\eta(t-h)$. This
equation and \eqref{eq:Iusedfilterign} completes the proof.
\end{proof}

The next corollary easily follows from Lemma~\ref{lem:key}.

\begin{cor}\label{cor:f_}
Let $i, j\in \{ 1, \dotsc, n\}$ and $t\geq T$ be arbitrary. Define $f_{ij} =
p_i^j  ((U_{ji} e^{\Pi^\top d_i^j})\otimes u_j^\top)$ and $g_{ij}(t) = q_i^j(t)
((U_{ji} e^{\Pi^\top t})\otimes u_j^\top)$. Then, for all $i$ and $t$, we have
\begin{align}
E[\eta_i(t)(\mathcal P_ix)(t)] 
&=
\sum_{j=1}^n f_{ij}\zeta(t-d_i^j),  \label{eq:E[delta_i(t)(P...}
\\
E[\eta_i(t)(Q_i * x)(t)] 
&=
\sum_{j=1}^n (g_{ij}*\zeta )(t).
\end{align}
\afterequation
\end{cor}

\begin{proof}
From the definition of the operator $\mathcal P_i$, we can show that
$\eta_i(t)(\mathcal P_i x)(t) = \sum_{j=1}^n u_j \eta_i(t)p_i^j x_j(t-d_i^j) =
\sum_{j=1}^n p_i^j u_j u_j^\top \eta_i(t) x(t-d_i^j)$. Taking the expectations
in the both hand sides of this equation, from Lemma~\ref{lem:key} we obtain
\begin{equation}
\begin{aligned}
E[\eta_i(t)(\mathcal P_i x)(t)]
&=
\sum_{j=1}^n p_i^j u_j u_j^\top \Bigl((u_i^\top e^{\Pi^\top d_i^j})\otimes I_n\Bigr)\zeta(t-d_i^j)
\\
&=
\sum_{j=1}^n p_i^j  \Bigl((u_ju_i^\top e^{\Pi^\top d_i^j})\otimes u_j^\top\Bigr)\zeta(t-d_i^j), 
\end{aligned}
\end{equation}
where we used $u_ju_j^\top = u_j\otimes u_j^\top$ and \eqref{eq:ABCD} to derive
the last equation. This equation proves \eqref{eq:E[delta_i(t)(P...}. We can
prove the other equation in the same way and hence omit its proof.
\end{proof}

Using Corollary~\ref{cor:f_}, we can then derive the dynamics of the variable
$\zeta$ as follows:

\begin{prop} \label{prop:barx}
 Define 
$\bar A_0 = \Pi^{\top}\otimes I_{n} + \bigoplus_{i=1}^n A_i$, 
$\bar A_{ij} = u_i\otimes f_{ij}$,
and $\bar B_{ij}(t) = u_i\otimes g_{ij}(t)$, 
for all $i, j\in \{1, \dotsc, n\}$ and $t\geq T$. Then, for every $t\geq T$, we have 
\begin{equation}
\frac{d\zeta}{dt}=\bar A_0\zeta(t)
+
\sum_{i,j=1}^n \bar A_{ij}\zeta(t-d_i^j)
+
\sum_{i,j=1}^n(\bar B_{ij}*\zeta)(t).\label{eq:barMJLSD}
\end{equation}
\afterequation
\end{prop}

\begin{proof}
We first derive a
differential equation for the extended state variable
\[
y=\begin{bmatrix}x\\
\eta
\end{bmatrix}.
\]
For each $i \in \{1, \dotsc, n \}$, define the operator $\mathcal A_i$ by $
(\mathcal A_i x)(t) = A_i x(t) + (\mathcal P_ix)(t) + (Q_i*x)(t)$. Then,
$\Sigma_1$ admits the representation $dx/dt = \mathcal A_{\epsilon(t)}x$.
Therefore, from the definition of the variables $\eta_i$, we can write $\Sigma_1$
as
\begin{equation}
\Sigma_1 : \frac{dx}{dt}=\sum_{i=1}^n\eta_i(\mathcal{A}_ix)(t). \label{eq:MJLSD.with.delta}
\end{equation}
Also, we know that $\eta$ follows the stochastic differential equation
\cite{Brockett2009} $d\eta=\sum_{i=1}^n \sum_{j\neq i}
(U_{ji}-U_{ii})\eta\,dN_{ij}$, where $N_{ij}$ denotes the Poisson counter of
rate $\pi_{ij}$ for each distinct pair~$(i,j) \in \{1, \dotsc, N \}^2$. This
equation and \eqref{eq:MJLSD.with.delta} show that
\[
dy=\begin{bmatrix}\sum_{i=1}^n\eta_i(\mathcal{A}_ix)\\
0
\end{bmatrix}dt+\sum_{i=1}^n \sum_{j\neq i}\begin{bmatrix}0\\
(U_{ji}-U_{ii})\eta
\end{bmatrix}dN_{ij}.
\]

Now, applying Lemma~\ref{lem:Ito} to the function $\psi(y) = \eta\otimes x = z$,
we obtain
\begin{equation*}
\begin{aligned}
dz
&=
\begin{multlined}[t][.85\linewidth]
\frac{\partial z}{\partial y}\begin{bmatrix}\sum_{i=1}^n\eta_i(\mathcal{A}_ix)\\
0
\end{bmatrix}\,dt+
\\
\sum_{i=1}^n \sum_{j\neq i}\left[\psi\left(y+\begin{bmatrix}0\\
(U_{ji}-U_{ii})\eta
\end{bmatrix}\right)-\psi(y)\right]dN_{ij}
\end{multlined}
\\
&=
\begin{multlined}[t][.85\linewidth]
\sum_{i=1}^n(\eta\otimes I_{n})\eta_i(\mathcal{A}_ix)\,dt 
+ 
\\
\sum_{i=1}^n \sum_{j\neq i}\left[\bigl((U_{ji}-U_{ii})\eta\bigr)\otimes x \right]\,dN_{ij},
\end{multlined}
\end{aligned}
\label{eq:d(del.ox.x).pre}
\end{equation*}
where we used the identity ${\partial \psi}/{\partial y}=[\eta\otimes I_{n} \
I_{n}\otimes x]$ in the last equation. Therefore, the expectation $\zeta$ obeys the differential equation
\begin{equation}\label{eq:dbarx/dt}
\begin{multlined}[.8\linewidth]
\frac{d\zeta}{dt} = 
\sum_{i=1}^nE[(\eta\otimes I_{n})\eta_i(\mathcal{A}_ix)] + \\
\sum_{i=1}^n \sum_{j\neq i}E\left[\bigl((U_{ji}-U_{ii})\eta\bigr)\otimes x \right]\pi_{ij}.
\end{multlined}
\end{equation}
Let us compute the expectations in the right hand side of this equation. Since
$\eta_i\eta_{j}=0$ for $i\neq j$ and $\eta_i^{2}=\eta_i$, we have
$\eta\eta_i=\eta_iu_i$. Therefore, it follows that $(\eta\otimes
I_{n})\eta_i(\mathcal{A}_ix) = (u_i\otimes A_i)\eta_ix+(u_i\otimes
I_n)\eta_i(\mathcal P_ix)+(u_i\otimes I_{n})\eta_i(Q_i*x)$. Hence, we can
compute the first term in the right hand side of \eqref{eq:dbarx/dt} as
\begin{equation*}
\begin{aligned}
&\sum_{i=1}^nE[(\eta\otimes I_{n})\eta_i(\mathcal{A}_ix)]
\\
=
&
\begin{multlined}[t][.85\linewidth]
\biggl(\bigoplus_{i=1}^n A_i\biggr) \zeta
+
\sum_{i=1}^n(u_i\otimes I)E[\eta_i(\mathcal P_ix)]
+\\
\sum_{i=1}^n(u_i\otimes I_{n})E[\eta_i(Q_i*x)]
\end{multlined}
\\
=
&
\biggl(\bigoplus_{i=1}^n A_i\biggr) \zeta
+
\sum_{i,j=1}^n \bar A_{ij}\zeta(t-d_i^j)
+
\sum_{i,j=1}^n(\bar B_{ij}*\zeta)(t), 
\end{aligned}
\end{equation*}
where we used Corollary~\ref{cor:f_} for deriving the last identity. On the
other hand, it is shown in the proof of \cite[Proposition~5.3]{Ogura2013f} that the second term of the right
hand side of \eqref{eq:dbarx/dt} equals $(\Pi^{\top}\otimes I_{n})\zeta$. This
completes the proof.
\end{proof}

We are now ready to prove Theorem~\ref{thm:rate}: 

\begin{proofof}{Theorem~\ref{thm:rate}}
Let $\phi$ and $\epsilon(0)$ be arbitrary. We first consider the special case of
$\lambda=0$. Assume that the matrix $\bar{T}^{(0)}$ is Hurwitz stable. Then, by
Proposition~\ref{prop:Ngoc}, the delayed positive linear
system~\eqref{eq:barMJLSD} is exponentially stable. Notice that the
equation~\eqref{eq:barMJLSD} is defined only for $t\geq T$. By the stability of
the system \eqref{eq:barMJLSD}, there exist $C_1 > 0$ and $\rho>0$ such that
\begin{equation}\label{eq:proof}
\norm{\zeta(t)} \leq C_1e^{-\rho(t-T)}\norm{\zeta\vert_{[0,T]}}. 
\end{equation}
On the other hand, due to the linearity of the system~$\Sigma_1$, there exists
$C_2>0$ such that $\norm{x\vert_{[0,T]}} \leq C_2 \norm \phi$. Using this
inequality, \eqref{eq:proof}, and Lemma~\ref{lem:presevrve}, we can show that
$E[\norm{x(t)}] \leq C_1 C_2 e^{-\rho(t-T)}\norm{\phi}$. This shows the exponential
stability of $\Sigma_1$.

For the general case, observe that the variable $\tilde x(t) = e^{-\lambda
t}x(t)$ satisfies the stochastic differential equation
\begin{equation}\label{eq:dx/dt^lambda}
\tilde \Sigma_1 : \frac{d\tilde x}{dt} = (A_{\epsilon(t)}-\lambda I) \tilde x(t)
+(\mathcal P^{(\lambda)}_{\epsilon(t)}\tilde x)(t) + (Q^{(\lambda)}_{\epsilon(t)}*\tilde x)(t), 
\end{equation}
where 
\begin{equation}
\begin{aligned}
(\mathcal P^{(\lambda)}_i\tilde x)(t) = \begin{bmatrix}
p_i^1 e^{-\lambda d_i^1} \tilde x(t-d_i^1)
\\
\vdots
\\
p_i^n e^{-\lambda d_i^n} \tilde x(t-d_i^n)
\end{bmatrix},\ 
Q_i^{(\lambda)}(t) = e^{-\lambda t}Q_i(t),
\end{aligned}
\end{equation}
for all $i\in \{1, \dotsc, n \}$ and $t\geq 0$. Applying the above argument on
exponential stability to $\tilde \Sigma_1$, we can show that $\tilde \Sigma_1$
is exponentially stable if $\bar T^{(\lambda)}$ is Hurwitz stable. This
completes the proof of the theorem because $\tilde \Sigma_1$ is exponentially
stable if and only if the growth rate of $\Sigma_1$ is less than $\lambda$.
\end{proofof}

\section{Growth Rate with Delayed Adaptation} \label{sec:delayedResponse}

In this section, we study the population model $\Sigma_2$ given in
\eqref{eq:dx_k/dt:resopsnedelay} for the case of delayed adaptation. We show
that we can characterize the growth rate of the populations as the maximum real
eigenvalue of a Metzler matrix, under the assumption that the delays are
described by a class of distributions called Coxian distributions. We focus on
the case $n=2$ for simplicity of presentations.

We consider the situation where the population as a whole updates its knowledge
$\sigma$ about the environment in the following stochastic manner:

\begin{enumerate}
\item When the environment changes from $i$ to $j$ at a time $t_0$ such that
$\sigma(t_0) = i$, a random number $T_{ij}$ is independently drawn from a
distribution $X_{ij}$. 

\item If the environment $\epsilon$ remains to be $j$ until the time $t_0+T_{ij}$,
then $\sigma(t)$ becomes $j$ at time $t_0+T_{ij}$.

\item If the value of $\epsilon$ changes before the time $t_0+T_{ij}$, then we
discard the number $T_{ij}$ and go back to the first step.
\end{enumerate}
In other words, if we let $t_1>t_0$ denote the next (minimum) time at which
$\epsilon$ changes, then we have
\begin{equation}
\sigma(t) = \begin{cases}
\sigma(t_0), &\text{$t_0\leq t\leq \min(t_0+T_{ij}, t_1)$,}
\\
\epsilon(t_0^+), &\text{$\min(t_0+T_{ij}, t_1) \leq t<t_1$.}
\end{cases}
\end{equation}
We call the distributions $X_{ij}$, or, the random times $T_{ij}$ as the
response delays.

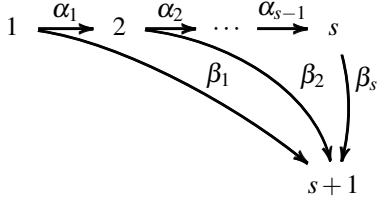
\begin{figure}[tb]
\newcommand{\myrule}{\rule[-.27cm]{0cm}{.7cm}}%
\centering
\centerline{\xymatrix@!C@=1.7em{
*+={\myrule 1} 
\ar[r]^-{\text{\normalsize$\alpha_1$}}
\ar@<.05ex>[r]
\ar@<-.05ex>[r]
\ar@<.075ex>[r]
\ar@<-.075ex>[r]
 \ar@(r,lu)[ddrrr]^-(.6){\text{\normalsize$\beta_1$}}
\ar@<.05ex>@(r,lu)[ddrrr]
\ar@<-.05ex>@(r,lu)[ddrrr]
\ar@<.075ex>@(r,lu)[ddrrr]
\ar@<-.075ex>@(r,lu)[ddrrr]
& 
*+={\myrule 2} 
\ar[r]^-{\text{\normalsize$\alpha_2$}} 
\ar@<.05ex>[r]
\ar@<-.05ex>[r]
\ar@<.075ex>[r]
\ar@<-.075ex>[r]
\ar@(r,u)[ddrr]^-(.665){\text{\normalsize$\beta_2$}}
\ar@<.05ex>@(r,u)[ddrr]
\ar@<-.05ex>@(r,u)[ddrr]
\ar@<.075ex>@(r,u)[ddrr]
\ar@<-.075ex>@(r,u)[ddrr]
& 
\ \cdots\ 
\ar[r]^-{\text{\normalsize$\alpha_{s-1}$}} 
\ar@<.05ex>[r]
\ar@<.075ex>[r]\ar@<-.075ex>[r]
\ar@<-.05ex>[r]
& 
*+={\myrule s} 
\ar@/^.2cm/[dd]^-(.255){\text{\normalsize$\beta_s$}}
\ar@<.05ex>@/^.2cm/[dd]
\ar@<-.05ex>@/^.2cm/[dd]
\ar@<.075ex>@/^.2cm/[dd]
\ar@<-.075ex>@/^.2cm/[dd]
\\
\\
&
&
&
*+={\myrule s+1}
}}
\caption{State transition diagram of the Coxian distribution describing the response delay $T_{12}$. $1$ is the entering state and $s+1$ is the absorbing state.}
\vspace{-3mm}
\label{fig:Coxian}
\end{figure}

We allow the response delays to follow a general class of distributions called
Coxian distributions defined as follows. For $\alpha \in \mathbb{R}_+^{s-1}$ and
$\beta\in\mathbb{R}_+^s$, consider the time-homogeneous Markov process having
the state transition diagram in Fig.~\ref{fig:Coxian}. We say that a random
variable follows the Coxian distribution (see, e.g., \cite{Asmussen1996}),
denoted by $C(\alpha,\beta)$, if it is the absorption time of the Markov process
into state $s+1$ starting from state $1$. It is known that the set of Coxian
distributions is dense in the set of positive valued
distributions~\cite{Cox1955}. Moreover, there are efficient fitting algorithms
to approximate a given arbitrary distribution by a Coxian
distribution~\cite{Asmussen1996}. We can now formally state our assumptions on
the response delays:

\begin{asm}
There exist $\alpha, \gamma \in \mathbb{R}_+^{s-1}$ and $\beta, \delta \in
\mathbb{R}_+^s$ such that $T_{12}$ and $T_{21}$ follow the Coxian distributions
$C(\alpha, \beta)$ and $C(\gamma, \delta)$, respectively.
\end{asm}

Combining the Markovian dynamics of the environment $\epsilon$ as well as the
state transition diagrams for the response delays $T_{12}=C(\alpha, \beta)$ and $T_{21} =
C(\gamma, \delta)$, we can easily prove the following proposition.

\begin{prop}\label{prop:isMP}
Consider the time-homogeneous Markov process $\theta$ having the state space
\begin{equation}
S = \{(1,1_0),(2,1_1),\dotsc (2,1_s), (2,2_0), (1,2_1), \dotsc, (1,2_s)\}
\end{equation}
and the state transition diagram in Fig.~\ref{fig:example:aggregatedMarkov}. Assume
that $\theta(0) = (\epsilon(0), \epsilon(0)_0)$. Define the function $f\colon S
\to \{1, 2\}\times \{1, 2\}$ by $f(i, j_k) = (i, j)$ for all $i, j\in \{1, 2\}$
and $k \in \{0, \dotsc, s\}$. Then, $f(\theta) = (\epsilon, \sigma)$.
\begin{figure}
\newcommand{\myrule}{\rule[-.27cm]{0cm}{.7cm}}%
\centering
\centerline{\xymatrix@!C@=1.7em{
*+={\myrule (1,1_0)} 
\ar@<-.5ex>[r]_-{\text{\normalsize$\pi_{12}$}}
&
*+={\myrule (2,1_1)} 
\ar[r]^-{\text{\normalsize$\alpha_1$}}
\ar@<.05ex>[r]
\ar@<-.05ex>[r]
\ar@<.075ex>[r]
\ar@<-.075ex>[r]
\ar@<-.5ex>[l]_(.4){\text{\normalsize$\pi_{21}$}} \ar@(r,lu)[ddrrr]^-(.6){\text{\normalsize$\beta_1$}}
\ar@<.05ex>@(r,lu)[ddrrr]
\ar@<-.05ex>@(r,lu)[ddrrr]
\ar@<.075ex>@(r,lu)[ddrrr]
\ar@<-.075ex>@(r,lu)[ddrrr]
& 
*+={\myrule (2,1_2)} 
\ar[r]^-{\text{\normalsize$\alpha_2$}} 
\ar@<.05ex>[r]
\ar@<-.05ex>[r]
\ar@<.075ex>[r]
\ar@<-.075ex>[r]
\ar@/^-.65cm/[ll]_(.4){\text{\normalsize$\pi_{21}$}}
\ar@(r,u)[ddrr]^-(.665){\text{\normalsize$\beta_2$}}
\ar@<.05ex>@(r,u)[ddrr]
\ar@<-.05ex>@(r,u)[ddrr]
\ar@<.075ex>@(r,u)[ddrr]
\ar@<-.075ex>@(r,u)[ddrr]
& 
\ \cdots\ 
\ar[r]^-{\text{\normalsize$\alpha_{s-1}$}} 
\ar@<.05ex>[r]
\ar@<-.05ex>[r]
\ar@<.075ex>[r]
\ar@<-.075ex>[r]
& 
*+={\myrule (2,1_s)} 
\ar@/^.2cm/[dd]^-(.255){\text{\normalsize$\beta_s$}}
\ar@<.05ex>@/^.2cm/[dd]
\ar@<-.05ex>@/^.2cm/[dd]
\ar@<.075ex>@/^.2cm/[dd]
\ar@<-.075ex>@/^.2cm/[dd]
\ar@/^-1.5cm/[llll]_(.4){\text{\normalsize$\pi_{21}$}}
\\
\\
*+={\myrule (1,2_s)} 
\ar@/^.2cm/[uu]^-(.315){\text{\normalsize$\delta_s$}}
\ar@<.05ex>@/^.2cm/[uu]
\ar@<-.05ex>@/^.2cm/[uu]
\ar@<.075ex>@/^.2cm/[uu]
\ar@<-.075ex>@/^.2cm/[uu]
\ar@/^-1.5cm/[rrrr]_(.4){\text{\normalsize$\pi_{12}$}}
&
\ \cdots\ 
\ar[l]^-{\text{\normalsize$\gamma_{s-1}$}}
\ar@<.05ex>[l]
\ar@<-.05ex>[l]
\ar@<.075ex>[l]
\ar@<-.075ex>[l]
&
*+={\myrule (1,2_2)} 
\ar[l]^-{\text{\normalsize$\gamma_2$}} 
\ar@<.05ex>[l] 
\ar@<-.05ex>[l]
\ar@<.075ex>[l] 
\ar@<-.075ex>[l]
\ar@(l,d)[lluu]^-(.7){\text{\normalsize$\delta_2$}}
\ar@<.05ex>@(l,d)[uull] 
\ar@<-.05ex>@(l,d)[uull]
\ar@<.075ex>@(l,d)[uull] 
\ar@<-.075ex>@(l,d)[uull]
\ar@/^-.65cm/[rr]_(.4){\text{\normalsize$\pi_{12}$}}
&
*+={\myrule (1,2_1)} 
\ar[l]^-{\text{\normalsize$\gamma_1$}} 
\ar@<.05ex>[l] 
\ar@<-.05ex>[l]
\ar@<.075ex>[l] 
\ar@<-.075ex>[l]
\ar@/^.75pc/[uulll]^(.6){\text{\normalsize$\delta_1$}}
\ar@<.05ex>@/^.75pc/[uulll] 
\ar@<-.05ex>@/^.75pc/[uulll]
\ar@<.075ex>@/^.75pc/[uulll] 
\ar@<-.075ex>@/^.75pc/[uulll]
\ar@<-.5ex>[r]_(.4){\text{\normalsize$\pi_{12}$}}
&
*+={\myrule (2,2_0)}
\ar@<-.5ex>[l]_-{\text{\normalsize$\pi_{21}$}}}}
\caption{Markov chain for the dynamics of the pair $(\epsilon, \sigma)$. The thick arrows represent the dynamics of phase-type distributions, while the thin arrows represent changes in the environment.}
\vspace{-3mm}
\label{fig:example:aggregatedMarkov}
\end{figure}
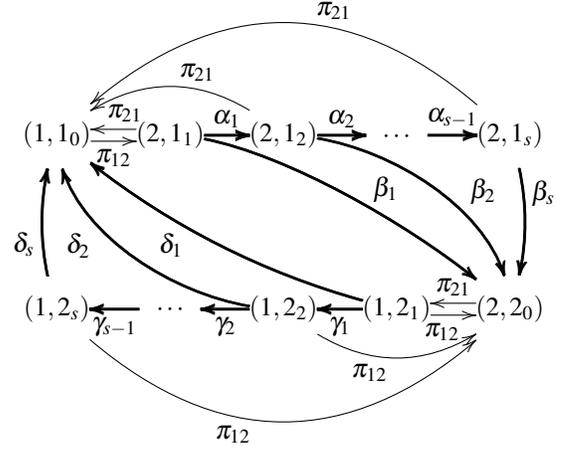
\end{prop}

From Proposition~\ref{prop:isMP}, we can represent the population dynamics $\Sigma_2$ as 
${dx}/{dt} = A_{f(\theta(t))}x(t)$, 
where, for each $i, j\in \{1, 2\}$, the matrix $A_{(i, j)}$ is defined by 
\begin{equation}\label{eq:def:A_(ij)}
A_{(i, j)} = \begin{bmatrix}
g_i^1 - \omega_j^{12} & \omega_j^{21}
\\
\omega_j^{12} & g_i^2 - \omega_j^{21}
\end{bmatrix}. 
\end{equation}
We now present the second main result of this paper, which gives the growth rate
of $\Sigma_2$:
 
\begin{thm}\label{thm:rate:responseDelay}
Define $\Xi\in\mathbb{R}^{(2s+2)\times (2s+2)}$ by
\begin{equation*}
\Xi = 
\begin{bmatrix}
-\pi_{12}&\pi_{12}u_1^\top&0&O_{1,s}
\\
\pi_{21}\onev_s&\!\!\!\!\!\!\Xi_\alpha-\pi_{21}I_s-\diag(\beta)\!\!\!\!\!\!\!\!\!\!&\beta&O_{s,s}
\\
O_{s,1}&O_{s,s}&-\pi_{21}&\pi_{21}u_1^\top
\\
\delta&O_{1,s}&\pi_{12}\onev_s&\!\!\!\!\!\!\! \Xi_\gamma-\diag(\delta)-\pi_{12}I_s
\end{bmatrix}, 
\end{equation*}
where
\begin{equation*}
\begin{aligned}
\Xi_\alpha &= \begin{bmatrix}
O_{s-1,1}&\diag(\alpha)
\\
0&O_{1,s-1}
\end{bmatrix} - \ \begin{bmatrix}
\diag(\alpha)&O_{s-1,1}
\\
O_{1,s-1}&0
\end{bmatrix}
\end{aligned}
\end{equation*}
and $\Xi_\gamma$ is defined in the same manner.
Then, the growth rate of $\Sigma_2$ equals the maximum real eigenvalue of the
matrix
\begin{equation}\label{eq:theBigMatrix}
\Xi^\top \otimes I_2+ \bigoplus(\bar A_{(1,1)}, I_s\otimes \bar A_{(2,1)}, \bar A_{(2,2)}, I_s\otimes \bar A_{(1,1)}).
\end{equation}
\afterequation
\end{thm}

\begin{proof}
It is easy to see that $\Xi$ defined in the theorem gives the infinitesimal
generator of the Markov process~$\theta$. Therefore, by
\cite[Theorem~5.2]{Ogura2013f}, the growth rate of $\Sigma_2$ equals
the maximum real eigenvalue of the Metzler matrix
\begin{equation}
\begin{multlined}[.85\linewidth]
\Xi^\top \otimes I_2 + \bigoplus(A_{f(1,1_0)}, A_{f(2,1_1)}, \dotsc, A_{f(2,1_s)},\\A_{f(2,2_0)}, A_{f(1,2_1)}, \dotsc, A_{f(1,2_s)}).
\end{multlined}
\end{equation}
The direct sum in this matrix equals the second term of \eqref{eq:theBigMatrix}
since $A_{f(2,1_k)} = A_{(2,1)}$ and $A_{f(1,2_k)} = A_{(1,2)}$ for all $k$ by
the definition of the matrices $A_{(i,j)}$ in \eqref{eq:def:A_(ij)}. This
completes the proof of the theorem.
\end{proof}

\section{Numerical Simulations \label{sec:Numerical-Simulations}}

In this section, we present numerical simulations to illustrate the results
obtained in the previous sections. For simplicity of presentation, we focus on
the case $n=2$; i.e., there are only two phenotypes in the population under
consideration. We use the parameters $g_1^1=1$, $g_1^2 = 0.05$, $g_2^1=-2$, and
$g_2^2 = 0.95$. These parameters indicate that the phenotypes 1 and 2 are fitted
to the environment 1 and 2, respectively. We set the phenotypic transition
rates as $\omega_1^{12}=\omega_{2}^{21} = 0.1$ and $\omega_1^{21} =
\omega_2^{12}= 1$.

First, we illustrate Theorem~\ref{thm:rate} for the case of delayed
proliferation. We consider only point delays; therefore, it is assumed that
$q_i^k(t)\equiv 0$ for all $i, k\in \{1, 2\}$. Furthermore, we assume that both
the delays and the rate of delayed proliferation are homogeneous, that is, there
exist $d\geq 0$ and $p\geq 0$ such that $d_i^k = d$ and $p_i^k = p$ for every
$i, k\in\{1,2  \}$. We set the initial state as $\phi(t)=[1\ 1]^\top$ for every
$t\in [-d,0]$. Using Theorem~\ref{thm:rate} and bisection search, we compute the
suboptimal upper bounds on the growth rates of~$\Sigma_1$ for $p =2.5$ and $d
\in [0, 5]$. To examine the accuracy of the upper bounds, we numerically compute
the quantity $50^{-1}\log(E[\norm{x(50)}]/\norm{\phi})$ using 500 sample paths
for each pair of $(d, p)$. The above two quantities are shown in
Fig~\ref{fig:staterates}. Their relative differences are less than 10\%, showing
the accuracy of the upper bounds by Theorem~\ref{thm:rate}. We have also
confirmed that, as $d\to\infty$ or $p\to0$, the upper bounds approach to the
common value $0.6863$, which equals the growth rate of the population model
$\Sigma_0$ without delays.

\begin{figure}[tb]
\vspace{1mm}
\centering
\includegraphics[trim=1cm 0 1cm 0,width=\linewidth]{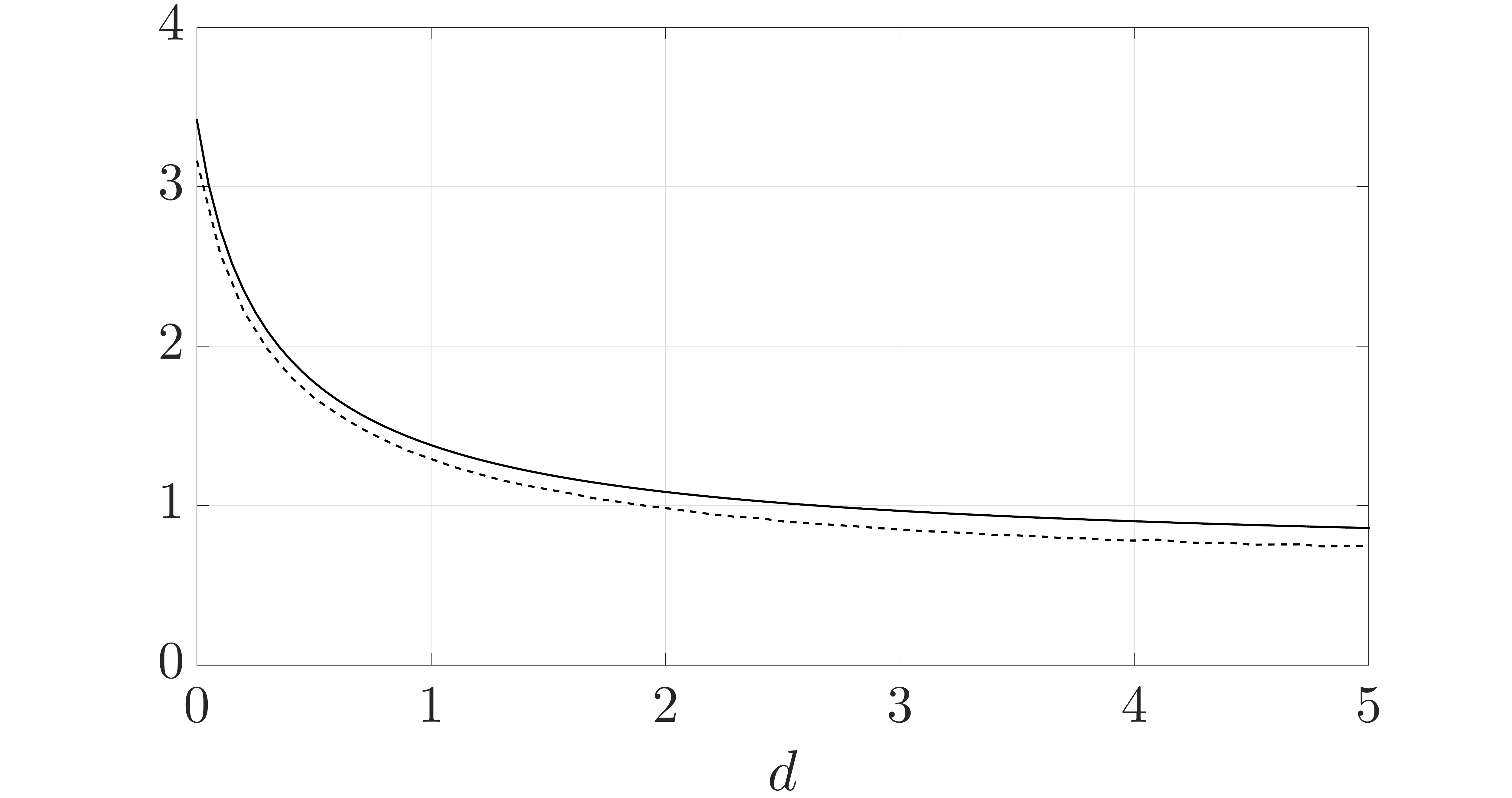}
\caption{Growth rates of $\Sigma_1$. Solid: Upper bounds (Theorem~\ref{thm:rate}). Dashed: Sample averages of $T^{-1}\log(\norm{x(T)}/\norm{\phi})$ with $T=50$.}
\label{fig:staterates}
\end{figure}

We then focus on delayed adaptation studied in
Section~\ref{sec:delayedResponse}. Assume that delays $X_{12}$ and $X_{21}$ both
follow the Erlang distribution with shape $k$ and mean $\mu$. This distribution
is the $k$-sum of independent exponential distributions with mean $\mu/k$ and,
therefore, approximates the normal distribution with mean $\mu$ and the variance
$\mu^2/k$ when $k$ is large. From this fact, we can also see that the Erlang
distribution is a Coxian distribution having the parameters $s=k$, $\alpha_1 =
\cdots = \alpha_{k-1} = \beta_{k} = \lambda = k/\mu$, and $\beta_1 = \cdots =
\beta_{k-1} = 0$.  Using Theorem~\ref{thm:rate:responseDelay}, we compute the
growth rate of $\Sigma_2$ when $\mu$ varies over the interval $[0, 10]$. We have
used $k=100$ in this simulation. We show the obtained growth rates in
Fig.~\ref{fig:growthrates}. We have confirmed the following limit phenomena.
First, as $\mu$ tends to zero, the growth rate approaches to that of the
population model $\Sigma_0$ without delay. Second, as $\mu$ tends to $\infty$,
the growth rate approaches to that of the population model without adaptation,
as expected.

\begin{figure}[tb]
\vspace{1mm}
\centering
\includegraphics[trim=1cm 0 1cm 0,width=\linewidth]{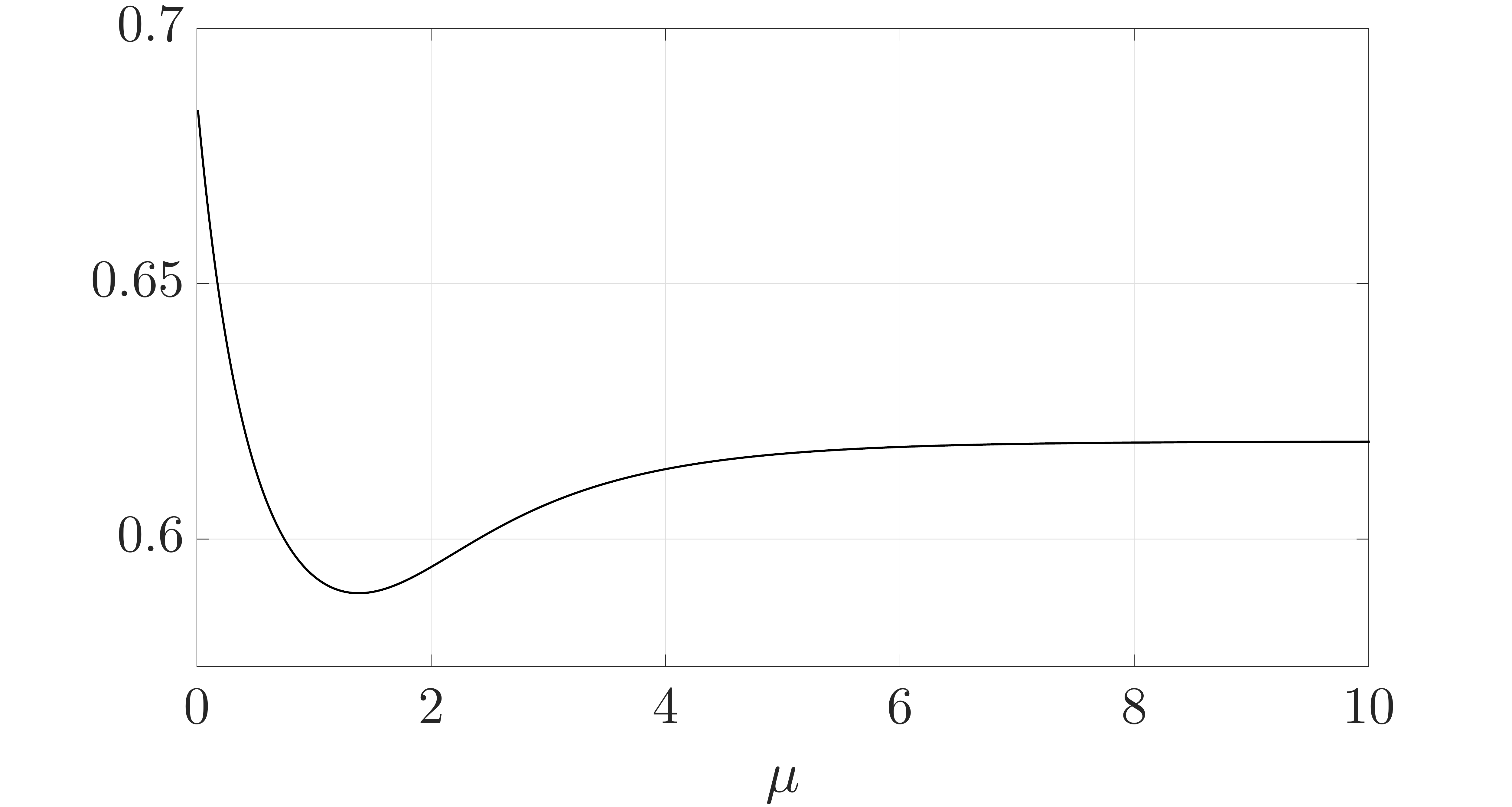}
\caption{Growth rate of $\Sigma_2$ versus $\mu$}
\vspace{-3mm}
\label{fig:growthrates}
\end{figure}

\section{Conclusion \label{sec:Conclusion}}

In this paper, we have studied the growth rate of bet-hedging populations
experiencing delays and environmental changes. By modeling the population
dynamics using positive Markov jump linear systems with delays, we have shown
that the growth rates can be upper-bounded by the maximum real eigenvalue of
Metzler matrices. In particular, in the case of adaptation delays, the upper
bounds give the exact value of the growth rates. We have confirmed the
effectiveness of the proposed methods via numerical simulations.

\end{document}